\documentclass[12 pt]{article}
\usepackage{amsmath}
\usepackage{amssymb}
\usepackage[a4paper, total={6in, 9in}]{geometry}
\usepackage{amsthm,url,tikz}
\usepackage[utf8]{inputenc}
\usepackage[english]{babel}
\usepackage{graphicx}

\usepackage[pagewise]{lineno}

\usepackage{hyperref}
\usepackage{adjustbox}
\numberwithin{equation}{section}
\usepackage{makecell}
\usepackage{bm}
\usepackage{xcolor}
\newtheorem{theorem}{Theorem}[section]
\newtheorem{proposition}[theorem]{Proposition}
\newtheorem{corollary}[theorem]{Corollary}
\newtheorem{lemma}[theorem]{Lemma}
\newtheorem{definition}[theorem]{Definition}

\newtheorem{example}[theorem]{Example}
\newtheorem{remark}[theorem]{Remark}
\usepackage{enumerate}
\usepackage[nottoc,notlot,notlof]{tocbibind} 

\title{ On Function-Correcting Codes in the Lee Metric  }
\author{Gyanendra K. Verma\footnote{Department of Electrical Communication Engineering,
 Indian Institute of Science, Bengaluru, Karnataka 560012, India. The author is financially supported by the ANRF-NPDF grant PDF/2025/002309. email: gkvermaiitdmaths@gmail.com}\ \  and 
 Abhay Kumar Singh\footnote{
 Department of Mathematics and Computing, Indian Institute of Technology (ISM) Dhanbad, Dhanbad.
 email: abhay@iitism.ac.in\ Corresponding author: Abhay Kumar Singh}}

\date{}

\begin{document}
	\maketitle
\begin{abstract} 
Function-correcting codes are a coding framework designed to minimize redundancy while ensuring that specific functions or computations of encoded data can be reliably recovered, even in the presence of errors. The choice of metric is crucial in designing such codes, as it determines which computations must be protected and how errors are measured and corrected. Previous work by Liu and Liu \cite{liu2025} studied function-correcting codes over  $\mathbb{Z}_{2^l},\  l\geq 2$ using the homogeneous metric, which coincides with the Lee metric over $\mathbb{Z}_4$. 
In this paper, we extend the study to codes over $\mathbb{Z}_m,$ for any positive integer $m\geq 2$ under the Lee metric and aim to determine their optimal redundancy. To achieve this, we introduce irregular Lee distance codes and derive upper and lower bounds on the optimal redundancy by characterizing the shortest possible length of such codes. These general bounds are then simplified and applied to specific classes of functions, including locally bounded functions, Lee weight functions, and Lee weight distribution functions. We extend the bounds established by Liu and Liu \cite{liu2025} for codes over $\mathbb{Z}_4$
in the Lee metric to the more general setting of $\mathbb{Z}_m$. Moreover, we give explicit constructions of function-correcting codes in Lee metric.
 Additionally, we explicitly derive a Plotkin-like bound for linear function-correcting codes in the Lee metric. As the Lee metric coincides with the Hamming metric over the binary field, we demonstrate that our bound naturally reduces to a Plotkin-type bound for function-correcting codes under the Hamming metric over $\mathbb{Z}_2$. 
 Furthermore, when the underlying function is bijective, function-correcting codes reduce to classical error-correcting codes. In parallel, our bound correspondingly reduces to the classical Plotkin bound for error-correcting codes, both for the Lee metric over $\mathbb{Z}_m$ and for the Hamming metric over $\mathbb{Z}_2$.


\end{abstract}

\section{Introduction}

The Lee metric was first introduced by Lee in 1958 as a way to measure distances over cyclic alphabets such as $\mathbb{Z}_m$. Later, Golomb and Welch\cite{golomb1968}, Berlekamp \cite{berlekamp1968}, and Chiang \cite{Chiang1971}  developed some of the earliest codes specifically designed for the Lee metric, recognizing its importance in modeling error patterns for non-binary and phase-based communication systems. A method for constructing Lee metric codes over arbitrary alphabet sizes using the elementary concepts of module theory was presented in \cite{satyanarayana1979}. These foundational works laid the groundwork for much of the modern research on coding over $\mathbb{Z}_m$ using the Lee metric. The Lee metric is often preferred over the Hamming metric in coding theory, particularly when working with codes over rings like $\mathbb{Z}_m$
  (instead of binary fields), because it better captures the geometry and error characteristics of non-binary alphabet. This motivates the study of function-correcting codes for the Lee metric, as codes optimized for this metric can leverage the modular structure of $\mathbb{Z}_m$
  to achieve lower redundancy and enhanced error correction performance for the same alphabet size. Several studies have shown that Lee metric based codes can yield tighter redundancy bounds compared to their Hamming metric counterparts, making them particularly effective for applications involving non-binary alphabets and structured, incremental error patterns.\\

Lenz et al.  \cite{Lenz2023} first introduced the framework of function-correcting codes (FCCs) in the Hamming metric, establishing the foundational theory for correcting function evaluations rather than raw symbols. This framework was later extended to other metrics to capture different error models. Specifically, Xia et al. \cite{Xia2023} generalized FCC to the symbol-pair metric, while Singh et al. \cite{Singh2025} extended it further to the $b$-symbol metric. Building on these developments, in \cite{Verma2025fcc}, the authors introduced locally $b$-symbol codes and investigated the redundancy of the corresponding function-correcting $b$-symbol codes. Several related advances complement these efforts: Ge et al. \cite{Ge2025} studied optimal redundancy bounds for such codes, Ly et al. \cite{Ly2025} proposed a conjecture concerning their asymptotic performance, and Premlal et al. \cite{Premlal2024} developed a function correcting framework for linear functions. In the context of locality, Rajput et al. \cite{Rajput2025} investigated locally function-correcting codes for the Hamming metric, while Sampath et al. \cite{Sampath2025} contributed further generalizations of FCC to specialized coding scenarios. Together, these works form a cohesive progression from classical FCC in the Hamming setting to rich variants supporting diverse metrics, locality, and function-specific error corrections.\\

Function-correcting codes with respect to the homogeneous metric were proposed in \cite{liu2025}, where the authors introduced function-correcting codes with homogeneous distance (FCCHDs) as a generalization of function-correcting codes under the Hamming metric. Since both Hamming weight and Lee weight over $\mathbb{Z}_4$ can be regarded as special cases of homogeneous weights, the results in \cite{liu2025} naturally extend to the Lee metric over $\mathbb{Z}_4$. Motivated by this fact, we investigate the construction and properties of function-correcting codes for the Lee metric over $\mathbb{Z}_m$, where $m\geq 2$ is an arbitrary positive integer. We derive upper and lower bounds on the optimal redundancy of function-correcting Lee metric codes for arbitrary functions by employing shortest-length irregular Lee distance codes in Section \ref{FCLMC}. In Section \ref{boundnD}, we provide several bounds on the shortest-length irregular Lee distance codes analogous to the Plotkin and Gilbert-Vashamov bounds. Furthermore, we introduce the concept of locally bounded Lee functions, and using general theoretical results, establish bounds on the optimal redundancy for this class of functions in Section \ref{locallyFCLMC}.  In particular, we demonstrate that the lower bound on optimal redundancy is achievable for locally binary Lee functions. Additionally, we investigate the locality properties of specific functions, the Lee weight function and the Lee weight distribution function, and determine bounds on their optimal redundancy in Section \ref{leeweight}. The bounds on the optimal redundancy of function-correcting codes can be obtained as a direct consequence of our general results. In particular, by substituting $m=2$ and 
$m=4$ into Lemma \ref{plotkin}, we get the specific cases corresponding to codes in the Hamming distance over $\mathbb{Z}_2$ in \cite{Lenz2023}  and the homogeneous distance over $\mathbb{Z}_4$ in \cite{liu2025}, respectively. This shows that the well-known bounds for these two settings follow naturally from our general framework, without requiring any additional arguments.  In Section \ref{linearfcc}, we investigate function-correcting codes associated with linear functions and derive a Plotkin-type bound on the optimal redundancy. We further demonstrate that this bound generalizes several existing Plotkin-like bounds, each of which emerges as a special case of our result.

\section{Preliminaries}\label{pre}
In this section, we discuss basic definitions and results that will be used throughout the paper. We denote the residue ring modulo $m$ by $\mathbb{Z}_m$, where $m$ is a positive integer greater than or equal to $2$. We consider the elements of $\mathbb{Z}_m$ as $0,1,\dots,m-1$ unless otherwise specified. For a positive integer $k$, $\mathbb{Z}_m^k$ denotes the set of all $k$-tuples over $\mathbb{Z}_m$.
\begin{definition}\cite[Lee distance metric]{Lee1958}
For $u=(u_1,u_2,\dots,u_k),v=(v_1,v_2,\dots,v_k)\in \mathbb{Z}_{m}^k$, the Lee distance between $u$ and $v$ is defined as 
$$d_L(u,v)=\sum_{i=1}^kw_L(u_i-v_i),$$
where $w_L(u_i-v_i)=\min\{u_i-v_i,v_i-u_i\}\pmod{m}$.
\end{definition}

If we denote the elements of $\mathbb{Z}_m$ as follows
\begin{align*}
    \begin{split}
        -\frac{m-1}{2},  -\frac{m-3}{2}, \dots,-1,0,1,\dots, \frac{m-1}{2} & \text{ if } m \text{ is odd}\\
        -\frac{m-2}{2},   \dots,-1,0,1,\dots, \frac{m-2}{2},\frac{m}{2} & \text{ if } m \text{ is even},
    \end{split}
\end{align*}
then the Lee weight of an element $x_i\in \mathbb{Z}_m$ is equivalently defined as (see \cite{Chiang1971})
\begin{equation}\label{alternatelee}
    w_L(x_i)=|x_i|.
\end{equation}

The set $\mathbb{Z}_m^k$ is called the message set. An encoding function $\mathcal{E}:\mathbb{Z}_m^k\to\mathbb{Z}_m^{k+r}$ is defined as $\mathcal{E}(u)=(u,p(u))$, where $p(u)\in\mathbb{Z}_m^r$ is the redundancy vector.
The set $$C=\{\mathcal{E}(u)\in\mathbb{Z}_m^{k+r}\mid u\in \mathbb{Z}_m^k\}$$ is called a code of length $k+r$ over $\mathbb{Z}_m$ with redundancy $r$, and $\mathcal{E}(u)$ is called a codeword in $C$. The minimum Lee distance of $C$ is defined as $$d_L(C)=\min\{d_L(\mathcal{E}(u),\mathcal{E}(v))\mid u,v\in \mathbb{Z}_m^k, u\neq v\}.$$ 
A ball centered at $u\in \mathbb{Z}_m^k$ of radius $\rho$ is defined as
$$B_L(u,\rho)=\{v\in \mathbb{Z}_m^k\mid d_L(u,v)\leq \rho\}.$$

The cardinality of a ball in $\mathbb{Z}_m$ with a given radius is independent of the choice of the center. We denote $V_L(r,\rho)= |\{x\in \mathbb{Z}_m^r: d_L(x,c)\leq \rho, \text{for some } c\in \mathbb{Z}_m^r\}|$, the cardinality of a ball of radius $\rho$ in $\mathbb{Z}_m^r$. As in the case of the Hamming metric, the sizes of balls appear in the Gilbert-Varshamov bound. Later, we derive a similar bound for function-correcting codes in the Lee metric. The following proposition is therefore useful. 

\begin{proposition}\cite[Proposition 10.10]{Roth1992}
  Let $\rho<m/2$. Then the size of a ball in $\mathbb{Z}_m^r$ is given by
  $$V_L(r,\rho)=\sum_{j=0}^r2^j\binom{r}{j}\binom{\rho}{j},$$
  where $\binom{\rho}{j}=0$ whenever $j>\rho$.
\end{proposition}

Let $r$ be a positive integer and $\rho$ be a real number. Denote 
\begin{align*}
B_r(\rho)=\{x\in \mathbb{Z}_m^r| \ w_L(x)\leq r\rho\}.
\end{align*}
In \cite{Loeliger1994},  Loeliger established an upper bound on $B_r(\rho)$ for any coordinate-additive metric, via the entropy of an auxiliary probability distribution. We state the bound for the Lee metric using \cite[Eq.(5)]{Loeliger1994}. 
\begin{lemma}\label{lemma-nmd}
    Let $0<\rho\leq 1$ be a real number. Then 
    \begin{align*}
     |B_r(\rho)|\leq (e^{\lambda\rho}\mathcal{L}(\lambda))^r,   
    \end{align*}
    holds for all non-negative integer $\lambda$ and $r$, where $\mathcal{L}(\lambda)=\sum_{a\in \mathbb{Z}_m}e^{-w_L(a)}$ and $e$ is natural constant.
\end{lemma}

\begin{lemma}\label{lemma-lambda}
 For $m\geq 5$,   we have $\sum_{a\in \mathbb{Z}_m}e^{-w_L(a)}< 3$. Moreover, for $m\leq 4$,  $\sum_{a\in \mathbb{Z}_m}e^{-w_L(a)}\leq 2$.
\end{lemma}
\begin{proof}
If $m=2\ell$ for some positive integer $\ell$, then
\begin{align*}
   \sum_{a\in \mathbb{Z}_m}e^{-w_L(a)}=e^0+e^{-\ell}+2\sum_{i=1}^{\ell-1}e^{-i}=1+e^{-\ell}+2(e^{-1}+e^{-2}+\dots+e^{-(\ell-1)})\\
   =1+e^{-\ell}+2\frac{e^{-1}(1-e^{-\ell+1})}{1-e^{-1}}=1+e^{-\ell}+2\frac{1-e^{-\ell+1}}{e-1}\\
   \leq 1+e^{-\ell}+2(1-e^{-\ell+1})=e^{-\ell}(1-2e)+3< 3.
\end{align*}
Similarly, if $m=2\ell+1$ for some positive integer $\ell$, then
\begin{align*}
   \sum_{a\in \mathbb{Z}_m}e^{-w_L(a)}=e^0+2e^{-\ell}+2\sum_{i=1}^{\ell-1}e^{-i}=1+2e^{-\ell}+2(e^{-1}+e^{-2}+\dots+e^{-(\ell-1)})\\
   =1+2e^{-\ell}+2\frac{e^{-1}(1-e^{-\ell+1})}{1-e^{-1}}=1+2e^{-\ell}+2\frac{1-e^{-\ell+1}}{e-1}\\
   \leq 1+2e^{-\ell}+2(1-e^{-\ell+1})=2e^{-\ell}(1-e)+3< 3.
\end{align*}
For $m\leq 4$, the bound is straightforward.
\end{proof}

\begin{lemma}\cite{Wyner1968}\label{sumS}
 Let $\alpha\in \mathbb{Z}_m$. Then \begin{align*}
     S=\sum_{\beta\in\mathbb{Z}_m}d_L(\alpha,\beta)=\begin{cases}
         \frac{m^2}{4}, & \text{ if } m \text{ is even }\\
         \frac{m^2-1}{4}, & \text{ if } m \text{ is odd. }
     \end{cases}
 \end{align*}   
\end{lemma}

\begin{lemma}\cite{Wyner1968}\label{Csum}
  Let $C=\{C_1,C_2,\dots,C_M\}\subseteq \mathbb{Z}_m^r$ be a code of length $r$ and size $M$. Then
  \begin{align*}
      \sum_{i,j}d_L(C_i,C_j)=\frac{SM^2r}{m}, \text{where $S$ is as in Lemma \ref{sumS}.}
  \end{align*}
\end{lemma}

\section{Function-correcting Lee metric codes}\label{FCLMC}
In error-correcting codes, we add redundancy to each message to enable the receiver to detect and correct any errors that affect the entire message during transmission. In contrast, in function-correcting codes, we need to correct messages with different function values. It allows for more efficient encoding. This approach reduces redundancy when several messages share the same function value. As a result, function-correcting codes offer a targeted and resource-efficient alternative to classical error correction. In this section, we formally define the function-correcting Lee metric code (in short, FCLMC) and establish an equivalence between the function-correcting Lee metric code and irregular-distance codes in the Lee metric. We discuss relations and bounds on the optimal redundancy of an $(f,t)$-FCLMC. Throughout the paper, we assume that a codeword $\mathcal{E}(u)$ is transmitted through a noisy channel, the received word is $y\in \mathbb{Z}_m^{k+r}$ with $d_L(\mathcal{E}(u),y)\leq t$. In this case, when the Lee distance between the transmitted code vector and received vector is at most $t$,  we say that at most $t$ errors occur during transmission, although it does not mean that $t$ coordinates changed during transmission. We define function-correcting Lee metric codes as follows.

\begin{definition}\label{def-fclmc}
  An encoding function $\mathcal{E}:\mathbb{Z}_m^k\to\mathbb{Z}_m^{k+r}  $ with $\mathcal{E}(u)=(u,p(u))$ defines an $(f,t)$-function-correcting Lee metric code (FCLMC) with redundancy $r$ if $d_L(\mathcal{E}(u),\mathcal{E}(v))\geq 2t+1$ holds for all $u,v\in \mathbb{Z}_m^k$ with $f(u)\neq f(v)$. 
\end{definition}

\begin{remark}
The encoding algorithm for function-correcting codes involves a systematic encoding, that is, the encoded message consists of the first $k$ coordinates corresponding to the message vector, followed by the concatenation of  $r$ coordinates of the redundancy vector. The process is done by an encoding function $\mathcal{E}:\mathbb{Z}_m^k\to\mathbb{Z}_m^{k+r}  $ with $\mathcal{E}(u)=(u,p(u))\in \mathbb{Z}_m^{k+r}$, where $p(u)$ is called the redundancy vector of length $r$. The collection $\{\mathcal{E}(u)| \ u\in \mathbb{Z}_m^k\}\subseteq \mathbb{Z}_m^{k+r}$  is called a code over $\mathbb{Z}_m$ of length $k+r$ (or equivalently, with redundancy $r$), and individual encoded messages $\mathcal{E}(u)$ are called codewords. Let $\mathcal{E}:\mathbb{Z}_m^k\to\mathbb{Z}_m^{k+r}$ define an $(f,t)$-FCLMC. For $u\in \mathbb{Z}_m^k$, $\mathcal{E}(u)\in \mathbb{Z}_m^{k+r}$ is transmitted through a noisy channel and receiver received $y\in \mathbb{Z}_m^{k+r}$. By Definition \ref{def-fclmc}, the encoding function $\mathcal{E}$ allows the receiver to determine the correct value of $f(u)$ whenever $d_L(\mathcal{E}(u),y)\leq t$ as follows.    \\
Consider set $S_y=\{v\in \mathbb{Z}_m^k|\ d_L(\mathcal{E}(v),y)\leq t\}$. The set is non-empty as $u\in S_y$. Then by triangle inequality, we have
\begin{align*}
    d_L(\mathcal{E}(u),\mathcal{E}(v))\leq d_L(\mathcal{E}(u),y)+d_L(y,\mathcal{E}(v))\leq 2t, \forall v\in S_y.
\end{align*}
Thus, by Definition \ref{def-fclmc}, $f(u)=f(v)$ for all $v\in S_y$. Therefore, to obtain correct value of $f(u)$, the receiver searches for  a vector $v\in \mathbb{Z}_m^k$ with $d_L(\mathcal{E}(v),y)\leq t$. Once such a vector $v$ is found, the receiver terminates the search and decodes $f(u)$ as $f(v)$.  Note that although $u\in S_y$, the found vector $v$ need not coincide with $u$, that is, under these conditions, the receiver may not  recover the exact message $u$.\end{remark}

\begin{definition}
The optimal redundancy $r_f^L(k,t)$ for a given function $f:\mathbb{Z}_m^k\to Im(f)$ is defined as 
$$r_f^L(k,t)=\min\{r\mid \exists  (f,t)-\text{FCLMC with redundancy } r\}.$$    
\end{definition}
Any classical error-correcting code with the minimum Lee distance $2t+1$, defines an $(f,t)$-FCLMC for any function $f:\mathbb{Z}_m^k\to Im(f)$.
For an integer $n\geq 1$, we denote $[n]=\{1,2,\dots,n\}$ and $[n]^{+}=\max\{n,0\}$. Define the distance matrix of a given function $f$ as follows.
\begin{definition}
  Let $u_1,u_2,\dots,u_M\in\mathbb{Z}_m^k$. Define the Lee distance requirement matrix $D_f^L(t,u_1,u_2,\dots,u_M)$ of order $M\times M$ in the Lee metric for a given function $f$ as follows:
  $$[D_f^L(t,u_1,u_2,\dots,u_M)]_{ij}=\begin{cases}
      [2t+1-d_L(u_i,u_j)]^{+} & \text{if } f(u_i)\neq f(u_j)\\
      0 & \text{ otherwise}
  \end{cases}
  $$
  for all $1\leq i,j\leq M$.
\end{definition}

The irregular-distance code in the Lee metric is defined as follows. 
\begin{definition}
    Let $D$ be an $M\times M$ square matrix over non-negative integers with diagonal entries zero. A code $C=\{C_1,C_2,\dots,C_M\}\subseteq\mathbb{Z}_m^r$ is  called a $D$-code in the Lee metric if $d_L(C_i,C_j)\geq [D]_{ij}$ for all $1\leq i,j\leq M$.
\end{definition}

Let $N_L(D)$ be the smallest integer $r$ such that there exists a $D$-code in the Lee metric of length $r$. When $[D]_{ij}=d$ for all $i\neq j$, we denote $N_L(M,d)$. Now, we prove a fundamental relationship between the redundancy of optimal FCLMCs and the irregular distance codes.

\begin{theorem}\label{rfnf}
Let $f:\mathbb{Z}_m^k\to Im(f)$ be a function and $\mathbb{Z}_m^k=\{u_1,u_2,\dots,u_{m^k}\}$.  Then 
$$r_f^L(k,t)=N_L(D_f^L(t,u_1,u_2,\dots,u_{m^k})).$$
\end{theorem}
\begin{proof}
Let $r_f^L(k,t)=r$ and $\mathcal{E}(u)=(u,p(u))$ be the corresponding encoding function. Let $C=\{p(u_1),p(u_2),\dots,p(u_{m^k})\}\subseteq \mathbb{Z}_m^r$ be a code of length $r$.\\
We claim that $C$ is a $D_f^L(t,u_1,u_2,\dots,u_{m^k})$-code, that is, $$N_L(D_f^L(t,u_1,u_2,\dots,u_{m^k}))\leq r.$$
If  $C$ is not a $D_f^L(t,u_1,u_2,\dots,u_{m^k})$-code, then there exist two distinct redundancy vectors $p(u_i)$ and $p(u_j)$ with $$0<d_L(p(u_i),p(u_j))<[D_f^L(t,u_1,u_2,\dots,u_{m^k})]_{ij}=2t+1-d_L(u_i,u_j).$$ This implies that $f(u_i)\neq f(u_j)$ and $$d_L(\mathcal{E}(u_i),\mathcal{E}(u_j))=d_L(u_i,u_j)+d_L(p(u_i),p(u_j))<2t+1,$$ 
which is a contradiction to the fact $\mathcal{E}$ defines an $(f,t)$-FCLMC.

Conversely, let $N_L(D_f^L(t,u_1,u_2,\dots,u_{m^k})) =N$  and $C=\{C_1,C_2,\dots,C_{m^k}\}$ be a $D_f^L(t,u_1,u_2,\dots,u_{m^k})$-code of length $N$. Define an encoding function $\mathcal{E}:\mathbb{Z}_m^k\to \mathbb{Z}_m^{k+N}$ as $\mathcal{E}(u_i)=(u_i,p(u_i))$, where $p(u_i)=C_i$ for all $1\leq i\leq m^k$. For $f(u_i)\neq f(u_j)$, we have 

\begin{align*}
    \begin{split}
        d_L(\mathcal{E}(u_i),\mathcal{E}(u_j))&=d_L(u_i,u_j)+d_L(C_i,C_j)\\
        &\geq d_L(u_i,u_j)+[D_f^L(t,u_1,u_2,\dots,u_{m^k})]_{ij}\\
       & \geq d_L(u_i,u_j)+2t+1-d_L(u_i,u_j)=2t+1.
    \end{split}
\end{align*}
Therefore, $\mathcal{E}$ defines an $(f,t)$-FCLMC. Thus $r_f^L(k,t)\leq N$. This completes the proof.
\end{proof}

\begin{corollary}\label{corndf}
 Let    $\{u_1,u_2,\dots u_M\}\subseteq\mathbb{Z}_m^k$ with $M\leq m^k$. Then the optimal redundancy of $(f,t)$-FCLMCs satisfies 
 $$r_f^L(k,t)\geq N_L(D_f^L(t,u_1,u_2,\dots,u_{M}))$$
 for any function $f:\mathbb{Z}_m^k\to Im(f)$.
\end{corollary}
\begin{proof}
 It is easy to see that    $$r_f^L(k,t)=N_L(D_f^L(t,u_1,u_2,\dots,u_{m^k}))\geq N_L(D_f^L(t,u_1,u_2,\dots,u_{M})),$$ since $M\leq m^k$. The rest follows from Theorem \ref{rfnf}.
\end{proof}

Next, we define the Lee distance between two distinct function values.
\begin{definition}\label{fncdistance}
    Let $f:\mathbb{Z}_m^k\to Im(f)$ be a function with $Im(f)=\{e_1,e_2,\dots,e_{\eta}\}$. Define Lee distance on $Im(f)$ as follows 
    $$d_L(e_i,e_j)=\min_{u_1,u_2\in\mathbb{Z}_m^k}\{d_L(u_1,u_2) |\  f(u_1)=e_i, f(u_2)=e_j\}.$$
\end{definition}
The function-distance matrix of a given function $f$ is defined as follows.
\begin{definition}
  Let $f:\mathbb{Z}_m^k\to Im(f)$ be a function with $Im(f)=\{e_1,e_2,\dots,e_{\eta}\}$. Define the function distance matrix with $(ij)$-th entry
  $$[D_f^L(t,e_1,e_2,\dots,e_{\eta})]_{ij}=\begin{cases}
      [2t+1-d_L(e_i,e_j)]^+ & \text{ if } i\neq j\\
      0 & \text{otherwise}.
  \end{cases}$$
\end{definition}

One possible approach for constructing  $(f,t)$-FCLMCs involves assigning an identical redundancy vector to all message vectors 
$u$ that correspond to the same function value under a given function. Although this approach is not strictly required for the construction of FCLMCs, it naturally gives an upper bound on the optimal redundancy of FCLMCs as follows. 

\begin{theorem}\label{imgredun}
 For any $f:\mathbb{Z}_m^k\to Im(f)$ with $Im(f)=\{e_1,e_2,\dots,e_{\eta}\}$, we have $$r_f^L(k,t)\leq N_L(D_f^L(t,e_1,e_2,\dots,e_{\eta})).$$
\end{theorem}
\begin{proof}
Let $N_L(D_f^L(t,e_1,e_2,\dots,e_{\eta}))=r$, $D_f^L(t,e_1,e_2,\dots,e_{\eta})=D$ and    $C=\{C_{e_i}\mid 1\leq i\leq \eta\}\subseteq \mathbb{Z}_m^r$ be a $D$-code in the Lee metric of length $r$.  Define an encoding function $\mathcal{E}(u)=(u, C_{f(u)})$.
We claim this encoding defines an $(f,t)$-FCLMC. For this, let $u,v\in \mathbb{Z}_m^k$ with $e_i=f(u)\neq f(v)=e_j$. This implies that  $C_{f(u)}\neq C_{f(v)}$ and
\begin{align*}
    \begin{split}
        d_L(\mathcal{E}(u),\mathcal{E}(v))=&d_L(u,v)+d_L(C_{f(u)},C_{f(v)})\\
        \geq& d_L(e_i,e_j)+ [D]_{ij}\\  
        \geq &d_L(e_i,e_j)+2t+1-d_L(e_i,e_j)\\
        =&2t+1.  
        \end{split}
\end{align*}
Hence $r_f^L(k,t)\leq r$.
\end{proof}

We have the following corollary from Theorem \ref{imgredun} and Corollary \ref{corndf}.
\begin{corollary}
Let $f:\mathbb{Z}_m^k\to Im(f)$ with $Im(f)=\{e_1,e_2,\dots,e_{\eta}\}$ and a set $u_1,u_2,\dots,u_{\eta}$ such that $f(u_i)=e_i$. If $D_f^L(t,u_1,u_2,\dots,u_{\eta})=D_f^L(t,e_1,e_2,\dots,e_{\eta})$, then 
$$r_f^L(k,t)=N_L(D_f^L(t,e_1,e_2,\dots,e_{\eta})).$$
\end{corollary}

\section{Bounds on the shortest length of irregular-distance codes}\label{boundnD}
In this section, we present some bounds on $N_L(D)$ that give bounds on the optimal redundancy using results from the previous section. In the next lemma, we give a lower bound on $N_L(D)$ for a matrix $D$ over non-negative integers, which is a generalization of the Plotkin bound for codes with irregular Lee distance requirements. The following lemma is proved in \cite[Theorem 3]{hareesh2026}. Here, we state and prove for the completeness.

\begin{lemma} \label{plotkin}
Let $D$ be an $M\times M$ matrix over non-negative integers with diagonal entries zero. Then 
  \begin{align*}
      N_L(D)\geq \begin{cases}
          \frac{8}{M^2m} \sum_{i,j; i<j}[D]_{ij}, & \text{ if } m \text{ is even}\\
          \frac{8m}{M^2(m^2-1)} \sum_{i,j; i<j}[D]_{ij}, & \text{ if } m\text{ is odd.}
      \end{cases}
  \end{align*}  
\end{lemma}
\begin{proof}
   Let $N_L(D)=r$. Let $C\subseteq\mathbb{Z}_m^r$ be a $D$-code in the Lee metric of length $r$ and size $M$. Consider an $M\times r$ matrix $A$ whose rows are codewords $\{C_1,C_2,\dots,C_M\}=C$. By \cite[Eq. (15)]{Wyner1968}, we have 
    \begin{align*}
    \begin{split}
        \sum_{1\leq i<j\leq M}d_L(C_i,C_j)&
        \leq \frac{SM^2r}{2m},
 \end{split}      
    \end{align*}
where \begin{align*}
    S=\begin{cases}
        \frac{m^2}{4}, & \text{ if } m \text{ is even}\\
        \frac{m^2-1}{4}, & \text{ if } m \text{ is odd.}
    \end{cases}
\end{align*}
Since $C$ is a $D$-code in the Lee metric, $d_L(C_i,C_j)\geq [D]_{ij}$. Thus,
\begin{align*}
   r\geq \frac{2m}{SM^2} \sum_{i,j; i<j}[D]_{ij}.
\end{align*}
This completes the proof.
\end{proof}

It is well known that the Lee metric coincides with the Hamming metric over $\mathbb{Z}_2$ and coincides with the homogeneous metric over $\mathbb{Z}_4$. Hence, we state the lemma for these two metrics as derived in \cite{Lenz2023} and \cite{liu2025}, respectively.
    
\begin{corollary}\cite{liu2025}\label{plot4}
  Let $m=4$.  For any matrix $D\in \mathbb{N}_0^{M\times M}$ with diagonal entries zero, we have 
     \begin{align*}
      N_L(D)\geq 
          \frac{2}{M^2} \sum_{i,j; i<j}[D]_{ij}.
  \end{align*} 
\end{corollary}

\begin{corollary}\cite{Lenz2023}\label{plot2}
    Let $m=2$.  For any matrix $D\in \mathbb{N}_0^{M\times M}$ with diagonal entries zero, we have 
     \begin{align*}
      N_L(D)\geq \frac{4}{M^2} \sum_{i,j; i<j}[D]_{ij}. 
  \end{align*} 
\end{corollary}

Analogous to the Gilbert-Varshamov bound, we derive the following lemma for irregular-distance codes in the Lee metric. 

\begin{lemma}\label{lemma-GVbound}
 Let $D\in \mathbb{N}^{M\times M}$ be a distance matrix. Then 
 $$N_L(D)\leq \min_{r\in \mathbb{N}}\left \{r:m^r>\max_{j\in [M]}\sum_{i=1}^{j-1}V_L(r,[D]_{ij}-1)\right \},$$
 where $V_L(r,\rho)$ denotes the cardinality of the ball in $\mathbb{Z}_m^r$ of radius $\rho$.
\end{lemma}
\begin{proof}
   We construct a $D$-code in the Lee metric of length $r$, where $$m^r\geq \max_{j\in [M]}\sum_{i=1}^{j-1}V_L(r,[D]_{ij}-1).$$ First, select an arbitrary codeword $C_1$. For $j=2$, since $m^r>V_L(r,[D]_{12}-1)$, there exists $C_2\in \mathbb{Z}_m^r$ such that $d_L(C_1,C_2)\geq [D]_{12}$. For $j=3$, $m^r>V_L(r,[D]_{13}-1)+V_L(r,[D]_{23}-1)$, there exists $C_3\in\mathbb{Z}_m^r$ such that $d_L(C_1,C_3)\geq [D]_{13}$ and $d_L(C_2,C_3)\geq [D]_{23}$. Similarly, we chose $C_4,C_5,\dots,C_M$ such that $d_L(C_i,C_j)\geq [D]_{ij}$ for $i<j$. Thus, if $$m^r>\max_{j\in [M]}\sum_{i=1}^{j-1}V_L(r,[D]_{ij}-1),$$ then there exist a $D$-code in the Lee metric, namely the code $\{C_1,C_2,\dots,C_M\}$. This completes the proof.
\end{proof}



Recall that $N_L(M, d)$ is the smallest possible length of a $D$-code in the Lee metric, where $D\in \mathbb{N}_0^{M\times M}$ and $[D]_{ij}=d$ for all $i\neq j$.  Using Lemma \ref{lemma-GVbound} and Lemma \ref{lemma-nmd}, we have an upper bound on $N_L(M,d)$.

\begin{corollary}\label{cor-GVm}
    For any $M,d,m\in \mathbb{N}$ with $M>1, m\geq 5$, we have
    \begin{align*}
        N_L(M,d)\leq \left \lceil \frac{\ln(M-1)+(d-1)}{\ln m-\ln 3}\right \rceil.
    \end{align*}
    Moreover, for $m=4$, we have \begin{align*}
        N_L(M,d)\leq \left \lceil \frac{\ln(M-1)+(d-1)}{\ln 2}\right \rceil.
    \end{align*}
\end{corollary}
\begin{proof}
Let $D$ be a matrix with $[D]_{ij}=d$ for all $i\neq j$. Then by Lemma \ref{lemma-GVbound},
\begin{align*}
    N_L(M,d)\leq \min_{r\in \mathbb{N}}\left \{r: \ m^r> \max_{j\in [M]}\sum_{i=1}^{j-1}V_L(r,d-1)    \right \}\\
    \leq \min_{r\in \mathbb{N}}\left \{r: \ m^r> (M-1)V_L(r,d-1)    \right \}.
\end{align*}
Take $\lambda=1$ in Lemma \ref{lemma-nmd}, we have 
\begin{align*}
    V_L(r,d-1)=\left | B_r\left (\frac{d-1}{r}\right ) \right |\leq e^{d-1}\left ( \sum_{a\in \mathbb{Z}_m} e^{-w_L(a)} \right )^r<e^{d-1}3^r.
\end{align*}
Consequently, $m^r>(M-1)V_L(r,d-1)$ whenever $r> \frac{\ln(M-1)+(d-1)}{\ln m-\ln 3}$. Hence, $N_L(M,d)\leq \left \lceil \frac{\ln(M-1)+(d-1)}{\ln m-\ln 3}\right \rceil$.
Further, using similar arguments, for $m=4$, we have
\begin{align*}
    V_L(r,d-1)=\left | B_r\left (\frac{d-1}{r}\right ) \right |\leq e^{d-1}\left ( \sum_{a\in \mathbb{Z}_m} e^{-w_L(a)} \right )^r<e^{d-1}2^r.
\end{align*}
Thus, $m^r>(M-1)V_L(r,d-1)$ whenever $r> \frac{\ln(M-1)+(d-1)}{\ln 2}$. We get the desired result. 
\end{proof}

\begin{remark}
The Lee metric coincides with the homogeneous metric for $m=4$ and with the Hamming metric for $m=2$.  In \cite{liu2025} and \cite{Lenz2023}, the authors derived bounds on $N_L(M,d)$ for $m=4$ and $m=2$, respectively. The bound in the above corollary for $m=4$ is tighter estimate compare to the corresponding bound obtained in Theorem 3.3 of \cite{liu2025}.
\end{remark}

Let
$N_L(\eta,d)=\min\{n\mid \exists \text{a code } C\subseteq\mathbb{Z}_m^n \text{ with } d_L(C)\geq d,  |C|=\eta\}$. That is, $N_L(\eta,d)$ is the minimum possible length of a code over $\mathbb{Z}_m$ with size $\eta$ and the minimum Lee distance $d$. The following theorem gives lower and upper bounds on the optimal redundancy of FCLMCs.

\begin{theorem}\label{redun}
Let $f:\mathbb{Z}_m^k\to Im(f)$ be a function with $\mid Im(f)\mid =\eta\geq 2$. Then
$$N_L(2,2t)\leq r_f^L(k,t)\leq  N_L(\eta,2t).$$
\end{theorem}
\begin{proof}
First, we prove that there exists $u,v\in \mathbb{Z}_m^k$ such that $f(u)\neq f(v)$ and $d_L(u,v)=1$. On contrary, suppose that for all $u,v\in \mathbb{Z}_m^k$ with $d_L(u,v)=1$, we have $f(u)=f(v)$. Fix, $u\in \mathbb{Z}_m^k$, then $f(v)=f(u)$ for all $v\in B_L(u,1)$. Next, fix a $u\neq u'\in B_L(u,1)$, then $f(v)=f(u')=f(u)$ for all $v\in B_L(u',1)$. Similarly, doing same process, we have $f(v)=f(u)$ for all $v\in \mathbb{Z}_m^k$. This implies $\mid Im(f)\mid=1$, which is a contradiction. Therefore, by Corollary \ref{corndf}, we have $r_f^L(k,t)\geq N_L(2,2t)$.

 Let $N_L(\eta,2t)=N$ and $C=\{C_i\mid 1\leq i\leq \eta\}\subseteq \mathbb{Z}_m^N$ be a code of length $N$ with size $\eta$ and minimum Lee distance $2t$. Let $Im(f)=\{e_1,e_2,\dots,e_\eta\}$.

 Define an encoding function $\mathcal{E}(u)=(u, C_u)$, where
 \begin{align*}
     C_u=\begin{cases}
         C_1 & \text{ if } f(u)=e_1\\
C_2 & \text{ if } f(u)=e_2\\
\vdots\\
C_{\eta} & \text{ if } f(u)=e_{\eta}.         
     \end{cases}
 \end{align*}
We claim that this encoding defines an $(f,t)$-FCLMC. For this, let $u,v\in \mathbb{Z}_m^k$ with $f(u)\neq f(v)$. Then $C_u\neq C_v$ and thus, 
\begin{align*}
    \begin{split}
        d_L(\mathcal{E}(u),\mathcal{E}(v))=d_L(u,v)+d_L(C_u,C_v)   \geq 1+2t.
        \end{split}
\end{align*}
Hence $r_f^L(k,t)\leq N_L(\eta,2t)$.
\end{proof}

\begin{remark}
For any $f:\mathbb{Z}_m^k\to Im(f)$ with $Im(f)=\{e_1,e_2,\dots,e_{\eta}\}$, by Theorem \ref{imgredun}, the optimal redundancy of FCLMCs is upper bounded by
$N_L(D_f^L(t,e_1,e_2,\dots,e_{\eta}))$. Note that \begin{align*}
    [D_f^L(t,e_1,e_2,\dots,e_{\eta})]_{ij}=[2t+1-d_L(e_i,e_j)]^+\leq 2t\ \  \text{ for all } i\neq j.
\end{align*}
Let $C$ be the code constructed in the proof of Theorem \ref{redun}. Since 
\begin{align*}
    d_L(C_i,C_j)\geq 2t\geq [D_f^L(t,e_1,e_2,\dots,e_{\eta})]_{ij},
\end{align*}
$C$ is a $D_f^L(t,e_1,e_2,\dots,e_{\eta})$-code in Lee metric. Thus, \begin{align*}
    N_L(D_f^L(t,e_1,e_2,\dots,e_{\eta}))\leq N_L(\eta,2t).
\end{align*}
Hence, Theorem $\ref{imgredun}$ gives a sharper bound than Theorem \ref{redun}. Nevertheless, Theorem \ref{redun}  remains significant as it gives an explicit bound together with an explicit construction (see Proposition \ref{prop-l2t}) that depends only on the cardinality of the image set and $t$. While Theorem \ref{imgredun} requires first computing the Lee distances between distinct function values in order to construct the matrix $D_f^L(t,e_1,e_2,\dots,e_{\eta})$ and subsequently determining the value $N_L(D_f^L(t,e_1,e_2,\dots,e_{\eta}))$.
\end{remark}

Theorem \ref{redun} suggests that determining the value $N_L(2,2t)$ and $N_L(\eta,2t)$ will give lower and upper bounds on the optimal redundancy of $(f,t)$-FCLMCs.

\begin{proposition}\label{nl2t}
For positive integers $t$ and $m\geq 2$, we have 
$$N_L(2,2t)=\left \lceil\frac{2t}{\left \lfloor \frac{m}{2}\right \rfloor} \right \rceil.$$
\end{proposition}
\begin{proof}
Denote $\left \lceil\frac{2t}{\left \lfloor \frac{m}{2}\right \rfloor} \right \rceil=s$. Consider $c=(0,0,,\dots,0)$ and $d=(\lfloor\frac{m}{2}\rfloor, \lfloor\frac{m}{2}\rfloor,\dots, \lfloor\frac{m}{2}\rfloor)$ of length $s$. Then
$$d_L(c,d)=\sum_{i=1}^sw_L\left (0-\left\lfloor \frac{m}{2} \right \rfloor\right)=s\cdot \left \lfloor \frac{m}{2} \right \rfloor\geq 2t .$$
Thus, $N_L(2,2t)\leq s$.

Conversely, let $c=(c_1,c_2,\dots,c_{s-1})$ and $d=(d_1,d_2,\dots,d_{s-1})$ be two codewords of length $s-1$. Then, as the maximum Lee weight of an element in $\mathbb{Z}_m$ can be $\left \lfloor\frac{m}{2}\right \rfloor$, we have
$$d_L(c,d)=\sum_{i=1}^{s-1}w_L (c_i-d_i)\leq (s-1)\cdot \left \lfloor \frac{m}{2} \right \rfloor <2t.$$
Therefore, $N_L(2,2t)\geq s$. This completes the proof. 
\end{proof}

In the following result, we construct a code over $\mathbb{Z}_m$ consisting of $\lambda$ codewords with minimum Lee distance $2t$ . This construction yields an upper bound on $N_L(\lambda, 2t)$ for arbitrary positive integers $\lambda$ and  $t$.
\begin{proposition}\label{prop-l2t}
Let $\lambda$ and $t$ be positive integers. Then
\begin{align*}
   N_L(\lambda,2t)\leq \lambda\left \lceil \frac{ t}{\left \lfloor\frac{m}{2}\right \rfloor}\right \rceil. 
\end{align*}
\end{proposition}
\begin{proof}
 Let $n=\left \lceil \frac{ t}{\left \lfloor\frac{m}{2}\right \rfloor}\right \rceil$, and $\bm{0}=(0,0,\dots,0)$, $\bm{a}=\left (  \left \lfloor \frac{m}{2}\right \rfloor, \left \lfloor \frac{m}{2}\right \rfloor,\dots, \left \lfloor \frac{m}{2}\right \rfloor   \right )\in \mathbb{Z}_m^{n}$. Consider \textcolor{red}{a} code $C=\{C_1,C_2,\dots,C_{\lambda}\}\subseteq \mathbb{Z}_m^{\lambda n}$, where for $1\leq i \leq \lambda$ is given by
 \begin{align*}
     C_i=(\bm{0}, \dots,\bm{0},\underbrace{\bm{a}}_{i\text{-th-place}},\bm{0},\dots,\bm{0})\in (\mathbb{Z}_m^{n})^\lambda.
 \end{align*}
 Observe that for $i\neq j$,
 \begin{align*}
     d_L(C_i,C_j)=2w_L(\bm{a})&=2n \left \lfloor \frac{m}{2}\right \rfloor=
     2\left \lceil \frac{ t}{\left \lfloor\frac{m}{2}\right \rfloor}\right \rceil  \left \lfloor \frac{m}{2}\right \rfloor\\
     & \geq 2\frac{ t}{\left \lfloor\frac{m}{2}\right \rfloor} \left \lfloor \frac{m}{2}\right \rfloor
    =2t,
 \end{align*}
 since $\lceil x\rceil\geq x$. Thus $C$ has the minimum Lee distance $2t$. This completes the proof.
\end{proof}

\begin{remark}
    By Proposition \ref{nl2t},  if $m=2$, then $N_L(2,2t)=2t$ and 
     if $m=4$, then $N_L(2,2t)=t$. Thus, Proposition \ref{nl2t} is consistent with the results obtained in \cite{Lenz2023} and \cite{liu2025} for the Hamming distance and the homogeneous distance, respectively. 
\end{remark}

\section{ FCLMCs for the locally functions}\label{locallyFCLMC}
In applications, we utilize specific functions or classes of functions that exhibit desirable properties, enhancing the efficiency of function-correcting codes. These structured functions often simplify encoding and decoding processes. In this section, we investigate one such class, known as locally defined functions. We explore the optimal redundancy of FCLMCs for these functions. This focused study helps in constructing more practical and effective FCLMCs. First, we provide values of $\lambda$ and $\rho$ for which a given function behaves like a locally $(\lambda,\rho)_L$-function, and then provide bounds on the optimal redundancy. Moreover, we describe locality and the optimal redundancy of $(f,t)$-FCLMCs for the Lee weight function and the Lee weight distribution function.  We start by defining the function ball and locally functions in the Lee metric.
\begin{definition}
    Let $u\in \mathbb{Z}_m^k$. A function ball centered at $u$ of radius $\rho$ is defined as
    $$B_L^f(u,\rho)=\{f(v)\mid v\in B_L(u,\rho)\}.$$
\end{definition}

\begin{definition}
    A function $f:\mathbb{Z}_m^k\to Im(f)$ is said to be a locally $(\lambda,\rho)_L$-function if $\mid B_L^f(u,\rho)\mid\leq \lambda$ for all $u\in \mathbb{Z}_m^k$.
\end{definition}
When $\lambda=2$, we say the $(\lambda,\rho)_L$-function is a locally binary Lee function.
As the optimal redundancy can be obtained by analyzing those message words for which function values are different.

Let $P$ be a nonempty set and let $\prec$ be a binary relation on $P$. We say that $\prec$ defines a total ordering on $P$ if the following conditions hold:
\begin{enumerate}
    \item for any $a,b \in P$, $a \prec b$ and $b \prec a$ imply $a=b$;
    \item whenever $a,b,c \in P$ satisfy $a \prec b$ and $b \prec c$, it follows that $a \prec c$;
    \item for all $a,b \in P$, the elements $a$ and $b$ are comparable, that is, either $a \prec b$, $b \prec a$, or $a=b$.
\end{enumerate}
A pair $(P,\prec)$ satisfying these properties is called a totally ordered set.
A subset $I \subseteq P$ is called a contiguous block if, for any two elements $a,b \in I$ with $a \prec b$, any $c \in P$ with $a \prec c \prec b$ is also contained in $I$.

Let $A$ be a finite set with $n$ elements. Write the elements of $A$ as
$a_0,a_1,\dots,a_{n-1}$. Define a relation $\prec$ on $A$ by
$a_i\prec a_j \quad \text{if and only if} \quad i\leq j$.
Then $\prec$ is a total order on $A$.
However, the contiguous block condition need not hold for an arbitrary subset of a totally ordered finite set. For example, consider the set
$$\mathbb{Z}_6=\{a_0=0,a_1=1,a_2=2,a_3=3,a_4=4,a_5=5\},$$
equipped with the above total order. Let
$I=\{a_0,a_4,a_5\} \subset \mathbb{Z}_6$.
Then $a_0\prec a_2\prec a_4$, but $a_2\notin I$. Hence, $I$ does not satisfy the contiguous block condition.

 The following lemma distinguishes the message words inside a function ball. The proof follows a similar approach to that of the Hamming distance \cite{Rajput2025}. Here, we provide the proof for completeness.

\begin{lemma}\label{tau}
Let $f:\mathbb{Z}_m^k\to Im(f)$ be a locally $(\lambda,\rho)_L$-function. Suppose there is a total order on the image set $\text{Im}(f)$ such that $B_L^f(u,\rho)$ forms a contiguous block for all $u\in \mathbb{Z}_m^k$. Then there exists a map $\tau_f^L:\mathbb{Z}_m^k\to [\lambda]$ such that for all $u,v\in\mathbb{Z}_m^k$,  if $f(u)\neq f(v)$ and $d_L(u,v)\leq \rho$, then  $\tau_f^L(u)\neq \tau_f^L(v)$. 
\end{lemma}
\begin{proof}
 Let $(\text{Im}(f),\prec)$ be a total order set and $B_L^f(u,\rho)$ form a contiguous block for all $u\in \mathbb{Z}_m^k$. Denote  $\text{Im}(f)=\{e_1\prec e_2\prec \dots, \prec e_\eta\}$ be the image set of $f$ ordered by $\prec$. Define a map $\gamma: \text{Im}(f)\to [\lambda]$ as $\gamma(e_j)=1+(j \pmod{\lambda})$.  Observe that for any contiguous block $I\subseteq \text{Im}(f)$ with $|I|\leq \lambda $, $\gamma$ is one-one on $I$.  Define map $\tau_f^L:\mathbb{Z}_m^k\to [\lambda]$, $\tau_f^L(u)=\gamma(f(u))$. For $u,v\in \mathbb{Z}_m^k$ with $f(u)\neq f(v)$ and $d_L(u,v)\leq \rho$, that is, $f(u),f(v)\in B_L^f(u,\rho)$. Since $|B_L^f(u,\rho)|\leq \lambda$, $\gamma$ is one-one on contiguous block $|B_L^f(u,\rho)|$. Thus $\tau_f^L(u)=\gamma(f(u))\neq \gamma(f(v))=
\tau_f^L(v)$. 
\end{proof}
In the following proposition, we show that the contiguous block condition in Lemma \ref{tau} is always satisfied for the Lee weight function over $\mathbb{Z}_m^k$. Similarly, we can prove for the Lee weight distribution function as well.

\begin{proposition}
  Let $f(u)=w_L(u)$ be the Lee weight function over $\mathbb{Z}_m^k$ and $\rho$ be a positive integer. Then $B_L^f(u,\rho)$ forms a contiguous block for all $u\in \mathbb{Z}_m^k$.  
\end{proposition}
\begin{proof} 
Let $f(u)=w_L(u)$ be Lee weight function over $\mathbb{Z}_m^k$. As $\text{Im}(f)\subset\mathbb{Z}$, let $\leq$ be  the natural order on $\text{Im}(f)$ induced from $\mathbb{Z}$. 
 Define a graph $G$ on $\mathbb{Z}_m^k$, $u,v\in \mathbb{Z}_m^k$ is connected if and only if $d_L(u,v)=1$. Trivially, for every $u\in \mathbb{Z}_m^k$, there exists $v\in \mathbb{Z}_m$ such that $d_L(u,v)=1$. Thus $G$ is connected. 
\\ \ \ \ For $u\in \mathbb{Z}_m^k$, let $B_L^f(u,\rho)\subseteq \text{Im}(f)$ be the Lee weight function ball of radius $\rho$ and centered at $u$. For $a,b\in B_L^f(u,\rho) $, let $c\in \text{Im}(f)$ with $a\leq c\leq b$. We claim that $c\in B_L^f(u,\rho)$, that is, $B_L^f(u,\rho)$ forms a contiguous block.
There exist $u_1,u_2\in \mathbb{Z}_m^k$ such that $w_L(u_1)=a$ and $w_L(u_2)=b$. Taking the shortest paths on $G$ from $u_1$ to $u$ and $u$ to $u_2$, we get
\begin{align*}
u_1=x_0,x_1,\dots,x_{d_1}=u=y_0,y_1,\dots,y_{d_2}=u_2   
\end{align*}
with $d_L(x_i,x_{i+1})=1$ and $d_L(y_j,y_{j+1})=1$. By triangle inequality,  $d_L(u_1,x_i)\leq i\leq d_L(u_1,u)\leq \rho$, and $d_L(u,y_j)\leq j\le  d_L(u,u_2)\leq \rho$, that is $x_i,y_j\in B_L(u,\rho)$ for all $0\leq i\leq d_1$ and $0\leq j\leq d_2$.
Rewrite the sequence as $u_1=v_0,v_1,\dots,v_n=u_2$. 
Then, we have  $|w_L(v_i)-w_L(v_{i+1})|\leq 1$, implies that $a=w_L(v_0),w_L(v_1),\dots,w_L(v_n)=b$ changes by at most one at each step.  This implies that every integer between $a$ and $b$ will be attained. To see, let there is a $c$ with $a<c<b$ and there is no $v_j$ such that $w_L(v_j)=c$. There are some $v_i$ and $v_j$ such that $w_L(v_i)<c$ and $w_L(v_j)>c$. Let $i'=\max\{i| \ w_L(v_i)<c\}$. Then $w_L(v_{i'+1})>c$, and consequently, $|w_L(v_i')-w_L(v_{i'+1})|\geq 2$, which a contradiction.
Thus, there is some $v_i$ such that $f(v_i)=c$ whenever $a\leq c\leq b$. This completes the proof. 
\end{proof}

Using Lemma \ref{tau}, we construct an  $(f,t)$-FCLMC for locally $(\lambda,2t)_L$-functions. This will give an upper bound on the optimal redundancy in terms of the existence of classical codes in the Lee metric with prescribed parameters depending on $\lambda$ and $t$.

\begin{theorem}\label{thmN2t}
    Let $f:\mathbb{Z}_{m}^k\to Im(f)$ be a locally $(\lambda,2t)_L$-function satisfying the contiguous block condition in Lemma \ref{tau}. Then, an upper bound on the optimal redundancy of an $(f,t)$-FCLMC is given by 
    $$r_f^L(k,t)\leq N_L(\lambda,2t).$$
\end{theorem}
\begin{proof}
By Lemma \ref{tau}, there exists a map $\tau_f^L:\mathbb{Z}_{m}^k\to [\lambda]$ such that  for $u,v\in \mathbb{F}_2^k$, $\tau_f^L(u)\neq \tau_f^L(v)$ whenever $f(u)\neq f(v)$ with $d_L(u,v)\leq 2t$. Let $C=\{C_1,C_2,\dots,C_{\lambda}\}$ be a code of length $N_L(\lambda,2t)$ with $\lambda$ codewords and minimum Lee distance $2t$. Define an encoding $\mathcal{E}:\mathbb{Z}_{m}^k\to \mathbb{Z}_{m}^{k+N_L(\lambda,2t)}$ given by $\mathcal{E}(u)=(u,C_{\tau_f^L(u)})$. We claim that this encoding defines an $(f,t)$-FCLMC. Let $u,v\in \mathbb{Z}_{m}^k$ with $f(u)\neq f(v)$. Then \\
\textbf{Case 1}: If $d_L(u,v)\geq 2t+1$, then 
$$d_L(\mathcal{E}(u),\mathcal{E}(v))\geq d_L(u,v)\geq 2t+1.$$\\
\textbf{Case 2}: If $d_L(u,v)\leq 2t$, then $\tau_f^L(u)\neq \tau_f^L(v)$. Thus $C_{\tau_f^L(u)}\neq C_{\tau_f^L(v)}$.  We have, 
\begin{align*}
    \begin{split}
d_L(\mathcal{E}(u),\mathcal{E}(v))=d_L(u,v)+d_L(C_{\tau_f^L(u)},C_{\tau_f^L(v)}) \geq 2t+1.   
    \end{split}
\end{align*}
Hence $r_f^L(k,t)\leq N_L(\lambda,2t)$.
\end{proof}
A well-known class of locally $(2,2t)_L$-functions satisfying the contiguous block condition in Lemma \ref{tau} is boolean-valued functions over $\mathbb{Z}_m^k$ defined as $f:\mathbb{Z}_m^k\to \mathbb{Z}_2$. The following corollary shows that the optimal redundancy of FCLMCs can be achieved for such locally binary Lee functions.

\begin{corollary}\label{cor-redundancy 2-locally}
    Let $f:\mathbb{Z}_m^k\to Im(f)$ be a locally $(2,2t)_L$-function satisfying the contiguous block condition in Lemma \ref{tau}. Then $$r_f^L(k,t)=\left \lceil\frac{2t}{\left \lfloor \frac{m}{2}\right \rfloor} \right \rceil.$$
\end{corollary}
\begin{proof}
Since $f$ is a locally $(\lambda,2t)_L$-function for $\lambda=2$, by Theorem \ref{thmN2t}, we get $r_f^L(k,t)\leq N_L(2,2t)$. Also, by Theorem \ref{redun}, $r_f^L(k,t)\geq N_L(2,2t)$. Thus, $r_f^L(k,t)=N_L(2,2t)$. From Proposition \ref{nl2t}, we have $N_L(2,2t)=\left \lceil\frac{2t}{\left \lfloor \frac{m}{2}\right \rfloor} \right \rceil$. This completes the proof.
\end{proof}

\section{ FCLMCs for the Lee weight and Lee weight distribution functions}\label{leeweight}
In this section, we explore the locality and optimal redundancy of FCLMCs for the Lee weight and Lee weight distribution functions. 
    Note that the Lee weight function and Lee weight distribution function satisfy the contiguous block condition in Lemma \ref{tau} with respect to the order induced from the set of integers $\mathbb{Z}$. Therefore, we can apply Theorem \ref{thmN2t} to establish bounds on optimal redundancy of FCLMCs for the Lee weight and Lee weight distribution functions. In addition, Proposition \ref{prop-l2t} gives a construction of FCLMCs for these functions, attaining the corresponding bounds.

Let $f(u)=w_L(u)$ for all $u\in \mathbb{Z}_m^k$ be Lee wight function. Recall that 
$$B_L(u,\rho)=\{v\in\mathbb{Z}_m^k\mid d_L(u,v)\leq \rho\}.$$

Let $v\in B_L(u,\rho)$. Then by the triangle inequality, we have
\begin{align*}
    \begin{split}
        w_L(v)=d_L(v,0)\leq d_L(v,u)+d_L(u,0)\leq \rho+w_L(u)
    \end{split}
\end{align*}
and 
\begin{align*}
    \begin{split}
        w_L(u)=d_L(u,0)\leq d_L(u,v)+d_L(v,0)\leq \rho+w_L(v).
    \end{split}
\end{align*}
Therefore, for any $v\in B_L(u,\rho)$, we have 
\begin{equation} \label{leew}
    w_L(u)-\rho \leq w_L(v)\leq w_L(u)+\rho.
\end{equation} 
This observation gives the following.
\begin{theorem}\label{thmwl}
 Let $\rho$ be a positive integer. Then the Lee weight function is a locally $(2\rho+1,\rho)_L$-function.    
\end{theorem}
\begin{proof}
   Let $f(u)=w_L(u)$. By Equation \ref{leew}, $\mid B_L^f(u,\rho)\mid\leq 2\rho+1$ for all $u\in \mathbb{Z}_m^k$.
\end{proof}

\begin{corollary}
    The optimal redundancy of an $(f,t)$-FCLMC for the Lee weight function $f(u)=w_L(u)$ satisfies $$r_f^L(k,t)\leq  (4t+1)\left \lceil \frac{ t}{\left \lfloor\frac{m}{2}\right \rfloor}\right \rceil.$$
\end{corollary}
\begin{proof}
By Theorem \ref{thmwl}, the Lee weight function $w_L(u)$ is a locally $(4t+1,2t)_L$-function. The rest follows from Theorem \ref{thmN2t} and Proposition \ref{prop-l2t}.
\end{proof}

We know that the maximum possible Lee weight of an element in $\mathbb{Z}_m$ can be $\left\lfloor\frac{m}{2} \right\rfloor$. In fact, for  any $u\in \mathbb{Z}_m^k$, the Lee weight of $u$, $w_L(u)\in\{0,1,\dots,k\left\lfloor\frac{m}{2}\right \rfloor\}=Im(w_L)$. It is easy to see that for any $i,j\in \{0,1,2,\dots,k\}=Im(w_H)$, where $w_H$ is the Hamming weight, there exist $u_i=(1^i,0^{k-i}),u_j=(1^j,0^{k-j})\in \mathbb{Z}_m^k$ such that $w_H(u_i)=i$, $w_H(u_j)=j$ and $d_H(u_i,u_j)=|i-j|$. In the Lee metric, finding such elements is non-trivial. In the following lemma, we construct such elements for the Lee metric.  

\begin{lemma}\label{ij}
 For all $i,j\in \{0,1,\dots,k\left\lfloor\frac{m}{2}\right \rfloor\}=Im(w_L)$, there exist $u_i,u_j\in \mathbb{Z}_m^k$ such that $w_L(u_i)=i$, $w_L(u_j)=j$ and $d_L(u_i,u_j)=|i-j|$.   
\end{lemma}
\begin{proof}
    Define 
$\mathcal{CL}_0=\{0\}$ and for $1\leq a\leq k$
$$\mathcal{CL}_a=\left\{i\in \mathbb{N}\mid (a-1)\left \lfloor \frac{m}{2}\right \rfloor < i\leq a \left \lfloor \frac{m}{2}\right \rfloor\right\}.$$
Observe that $$\bigcup_{0\leq a\leq k}\mathcal{CL}_a=Im(w_L) .$$
Also, let $u_0=(0,0,\dots,0)\in \mathbb{Z}_m^k$. For $1\leq i\leq k\left \lfloor \frac{m}{2}\right \rfloor$, if $i\in \mathcal{CL}_a$, that is, $(a-1)\left \lfloor \frac{m}{2}\right \rfloor < i\leq a \left \lfloor \frac{m}{2}\right \rfloor$, then

$$u_i=\left (\left\lfloor\frac{m}{2}\right \rfloor, \left\lfloor\frac{m}{2}\right \rfloor, \dots, \left\lfloor\frac{m}{2}\right \rfloor, i-(a-1)\left\lfloor\frac{m}{2}\right \rfloor (a^{\mathrm{th}}\text{-place}), 0, 0,\dots,0  \right)\in \mathbb{Z}_m^k.$$

Since $(a-1)\left \lfloor \frac{m}{2}\right \rfloor < i\leq a \left \lfloor \frac{m}{2}\right \rfloor$, by equation (\ref{alternatelee}), we have
 $$w_L\left (i-(a-1)\left\lfloor\frac{m}{2}\right \rfloor\right )=i-(a-1)\left\lfloor\frac{m}{2}\right \rfloor.$$ 

Consequently,
$$w_L(u_i)=
    (a-1)\left\lfloor\frac{m}{2}\right \rfloor+i-(a-1)\left\lfloor\frac{m}{2}\right \rfloor=i.$$

Now, we show that $d_L(u_i,u_j)=|i-j|$.  For this,\\
\textbf{Case 1}: Let $i,j\in \mathcal{CL}_a$ for some $a$. Without loss of generality, assume that $j> i$. Then
\begin{align*}
    d_L(u_i,u_j)=&w_L\left (0,0,\dots, i-(a-1)\left\lfloor\frac{m}{2}\right \rfloor-(j-(a-1)\left\lfloor\frac{m}{2}\right \rfloor   (a^{\mathrm{th}}\text{-place}),0,\dots,0 \right)\\
    =&w_L(0,0,\dots,i-j (a^{\mathrm{th}}\text{-place}),0,\dots,0)\\
    =&w_L(i-j).
\end{align*} 
Observe that $-\left\lfloor\frac{m}{2}\right \rfloor\leq i-j <0$, thus by equation (\ref{alternatelee}), $w_L(i-j)=|i-j|$.
Hence $d_L(u_i,u_j)=|i-j|.$

\textbf{Case 2}: Let $i\in \mathcal{CL}_a$ and $j\in \mathcal{CL}_b$ for some $a\neq b$. Without loss of generality, let $j>i$, then $b>a$ and 
\small \begin{align*}
     &u_i-u_j=\\
     & (0,0,\dots,0, \underbrace{i-(a-1)\left\lfloor\frac{m}{2}\right \rfloor-\left\lfloor\frac{m}{2}\right \rfloor}_{ (a^{\mathrm{th}}\text{-place})}, -\left\lfloor\frac{m}{2}\right \rfloor,\dots,-\left\lfloor\frac{m}{2}\right \rfloor,\underbrace{-j+(b-1)\left\lfloor\frac{m}{2}\right \rfloor}_{(b^{\mathrm{th}}\text{-place})},0,\dots,0  ).
 \end{align*}
 Thus 
 \begin{align*}
     d_L(u_i,u_j)=&w_L(u_i-u_j)= w_L\left (i-a\left\lfloor\frac{m}{2}\right \rfloor\right)+ (b-a-1)w_L\left (-\left\lfloor\frac{m}{2}\right \rfloor\right )+w_L\left (-j+(b-1)\left\lfloor\frac{m}{2}\right \rfloor\right ).    
 \end{align*}
 Note that $(a-1)\left\lfloor\frac{m}{2}\right \rfloor<i\leq a\left\lfloor\frac{m}{2}\right \rfloor$. This  implies that $-\left\lfloor\frac{m}{2}\right \rfloor\leq i-a\left\lfloor\frac{m}{2}\right \rfloor\leq0\leq \left\lfloor\frac{m}{2}\right \rfloor$. Therefore, again by equation (\ref{alternatelee}), $$ w_L(i-a\left\lfloor\frac{m}{2}\right \rfloor)= \left | i-a\left\lfloor\frac{m}{2}\right \rfloor \right |=-i+a\left\lfloor\frac{m}{2}\right \rfloor$$

 \begin{align*}
     d_L(u_i,u_j)=&w_L(u_i-u_j)= w_L\left (i-a\left\lfloor\frac{m}{2}\right \rfloor\right )+ (b-a-1)w_L\left (\left\lfloor\frac{m}{2}\right \rfloor\right )+w_L\left (j-(b-1)\left\lfloor\frac{m}{2}\right \rfloor\right )\\
     =&-i+a\left\lfloor\frac{m}{2}\right \rfloor+b\left\lfloor\frac{m}{2}\right \rfloor-a\left\lfloor\frac{m}{2}\right \rfloor-\left\lfloor\frac{m}{2}\right \rfloor+j-b\left\lfloor\frac{m}{2}\right \rfloor+\left\lfloor\frac{m}{2}\right \rfloor\\
     =&-i+j=|i-j|\ \ \ \text{ (since }  i-j<0, |i-j|=-(i-j)=-i+j ).
 \end{align*}
 This completes the proof.
\end{proof}

Using the above lemma, we have the following.

\begin{theorem}\label{leeweghtnd}
    Let $f(u)=w_L(u)$ be the Lee weight function with $Im(w_L)=\{0,1,\dots,k\left\lfloor\frac{m}{2}\right \rfloor=\eta\}$. Let $D_{w_L}(t)$ be an $(\eta+1)\times (\eta+1)$ matrix with $(i,j)$-th entry
    $$[D_{w_L}(t)]_{ij}=\begin{cases}
        0 & \text{ if } i=j\\
 [2t+1-|i-j|]^+ & \text{ if } i\neq j.
 \end{cases}$$
 Then $r_f^L(k,t)=N_L(D_{w_L}(t))$.
\end{theorem}
\begin{proof}
 Let $e_i=i$ for all $i\in Im(w_L)$.  From Definition \ref{fncdistance}, we have 
    $$d_L(e_i,e_j)=\min_{u_1,u_2\in\mathbb{Z}_m^k}\{d_L(u_1,u_2) |\  f(u_1)=e_i, f(u_2)=e_j\}.$$ 
    Using the triangle inequality, $d_L(e_i,e_j)\geq |i-j|$. Also, by Lemma  \ref{ij}, there exist $u_1,u_2\in \mathbb{Z}_m^k$, such that $w_L(u_1)=i$, $w_L(u_2)=j$ and $d_L(u_i,u_j)=|i-j|$. Hence $d_L(e_i,e_j)=|i-j|$. By Theorem \ref{imgredun}, we have
    $$r_f^L(k,t)\leq N_L(D_{w_L}(t)).$$
    Conversely, for all $i\in \left \{0,1,\dots,k\left\lfloor\frac{m}{2}\right \rfloor\right \}$, by Lemma \ref{ij}, there exits $u_i\in \mathbb{Z}_m^k$ with $w_L(u_i)=i$ such that $d_L(u_i,u_j)=|i-j|$. Then by Corollary \ref{corndf}, $r_f^L(k,t)\geq N_L(D_{w_L}(t))$.    
\end{proof}

\begin{example}
  Let $m=4$, $t=2$ and $k=4$. Then $\eta=k\left \lfloor\frac{m}{2}\right \rfloor=8$ and $D_{w_L}(2)$ is a $9\times 9$ matrix given by
  $$\begin{bmatrix}
  
      0 &4&3&2&1&0&0&0&0\\
      4&0&4&3&2&1&0&0&0\\
      3&4&0&4&3&2&1&0&0\\
      2&3&4&0&4&3&2&1&0\\
      1&2&3&4&0&4&3&2&1\\
      0&1&2&3&4&0&4&3&2\\
      0&0&1&2&3&4&0&4&3\\
      0&0&0&1&2&3&4&0&4\\
      0&0&0&0&1&2&3&4&0
  \end{bmatrix}.$$
\end{example}
Using Theorem \ref{leeweghtnd} and Lemma \ref{plotkin}, we get a lower bound on the optimal redundancy of FCLMCs for the Lee weight function.
\begin{corollary}\label{wlower}
Let $f(u)=w_L(u)$ and $\eta=k\left \lfloor \frac{m}{2}\right \rfloor$. For any $k\geq\left \lceil \frac{t+1}{\left \lfloor \frac{m}{2}\right \rfloor}\right \rceil$,
\begin{align*}
  r_f^L(k,t)\geq \begin{cases}
      \left \lceil \frac{10t^3+30t^2+20t}{3(t+2)^2\left \lfloor \frac{m}{2}\right \rfloor} \right \rceil & \text{ if $m$ is even,}\\\\
      \left \lceil \left (\frac{m}{m^2-1}\right ) \left (\frac{20t^3+60t^2+40t}{3(t+2)^2}\right ) \right \rceil & \text{ if $m$ is odd}.
  \end{cases}  
\end{align*}
\end{corollary}
\begin{proof}
 Let $\{C_1,C_2,\dots,C_{\eta+1}\}$ be a $D_{w_L}(t)$-code. Consider first $t+2$ codewords $C_1,C_2,\dots,C_{t+2}$. Let $m$ be even. Then  by Lemma \ref{plotkin} and Theorem \ref{leeweghtnd}, 
\begin{align*}
    r_f^L(k,t)=&N_L(D_{w_L}(t))\geq  \frac{4}{(t+2)^2\left \lfloor \frac{m}{2}\right \rfloor} \sum_{1\leq i,j\leq t+2; i<j}[D_{w_L}(t)]_{ij}\\
    =&\frac{4}{(t+2)^2\left \lfloor \frac{m}{2}\right \rfloor} \sum_{i=1}^{t+2}\sum_{j=i+1}^{t+2}[D_{w_L}(t)]_{ij}\\
    =&\frac{4}{(t+2)^2\left \lfloor \frac{m}{2}\right \rfloor} \left (\sum_{i=0}^{t}(2t-i)(t+1-i)\right )\\
    =& \frac{4}{(t+2)^2\left \lfloor \frac{m}{2}\right \rfloor} \frac{5t^3+15t^2+10t}{6}\\
    =& \frac{10t^3+30t^2+20t}{3(t+2)^2\left \lfloor \frac{m}{2}\right \rfloor}.
\end{align*} 
Similarly, using Lemma \ref{plotkin} and Theorem \ref{leeweghtnd}, we get the desired result when $m$ is odd. Since $r_f^L(k,t)$ is an integer, this completes the proof.
\end{proof}

\begin{remark}
    Let $f:\mathbb{Z}_4^k\to Im(f)$ be the Lee weight function. Then by Corollary \ref{wlower}, for any $k\geq \left \lceil \frac{t+1}{2}\right \rceil$,
$$r_f^L(k,t)\geq \frac{5t^3+15t^2+10t}{3(t+2)^2}.$$
In \cite[Corollary 4.2]{liu2025}, the authors proved the same bound for the homogeneous metric. As the homogeneous metric and the Lee metric are the same over $\mathbb{Z}_4$. However, Corollary \ref{wlower} gives the bound over $\mathbb{Z}_m$ for any positive integer $m$ with respect to the Lee metric. 
\end{remark}

For $a, b\in \mathbb{Z}$, define shifted modulo \cite{Lenz2023} as follows
\begin{align*}
    a \text{ smod }b \equiv (a-1)\pmod{b}+1\in \{1,2,\dots,b\}.
\end{align*}
For example, $4  \text{ smod }4 \equiv 4$ and $5 \text{ smod } 4 \equiv 1$.  

Now, we give a construction of FCLMCs for the Lee weight function over $\mathbb{Z}_m^k$. that attains the optimal redundancy in Corollary \ref{wlower} for several cases.\\

\noindent
\textbf{Construction of FCLMC for the Lee weight function}\\
For $t\geq 1$, let $p_1,p_2,\dots,p_{2t+1}$ be a code with the minimum Lee distance $2t$, that is, $d_L(p_i,p_j)\geq 2t$, and $p_i=p_{i \text{ smod } (2t+1)}$ for $i\geq 2t+2$. Define an encoding $\mathcal{E}$ over $\mathbb{Z}_m^k$ given by $\mathcal{E}(u)=(u,p_{w_L(u)+1})$. For $u,v\in \mathbb{Z}_m^k$ with $w_L(u)\neq w_L(v)$, let $w_L(u)=w_1 \text{ smod } (2t+1)$, $w_L(v)=w_2 \text{ smod } (2t+1)$.\\
\textbf{Case 1} If $|w_L(u)-w_L(v)|\geq 2t+1$, then
\begin{align*}
d_L(\mathcal{E}(u),\mathcal{E}(v))\geq d_L(u,v)=w_L(u-v)\geq |w_L(u)-w_L(v)|\geq 2t+1.
\end{align*}
\textbf{Case 2} If  $|w_L(u)-w_L(v)|< 2t+1$, then $w_1\neq w_2$, that is, $p_{w_L(u)+1}\neq p_{w_L(v)+1}$. Thus, we have
\begin{align*}
d_L(\mathcal{E}(u),\mathcal{E}(v))= d_L(u,v)+d_L(p_{w_L(u)+1}, p_{w_L(v)+1})\geq 2t+1.
\end{align*}
Hence, by Definition \ref{def-fclmc}, encoding $\mathcal{E}$ defines an $(f,t)$-FCLMC for the Lee weight function over $\mathbb{Z}_m^k$.
In the following example, we construct $p_i$'s for various values of $t$.
\begin{example}
(1) For $m\geq 4$ and $t=1$, let $p_1=(0,0)$, $p_2=(1,1)$, $p_3=\left ( \left \lfloor \frac{m}{2} \right \rfloor,0 \right )$, and $p_i=p_{i \text{ smod } 3}$ for $i\geq 4$. Moreover, when $m=4$, the bound in Corollary \ref{wlower} is tight. \\
(2) For $m\geq 5$ and $t=2$, let $p_1=(0,0,0) , p_2=(2,2,0) , p_3=(2,0,2) , p_4=(0,2,2) ,p_5=(3,3,3)$, and $p_i=p_{i \text{ smod } 5}$ for $i\geq 6$. Moreover, when $m=5$, the bound in Corollary \ref{wlower} is tight.
\end{example}

Now, we describe the locality of the Lee weight distribution function. The Lee weight distribution function with a threshold $T>0$ is defined as 
$$W_L^T(u)=\left \lfloor \frac{w_L(u)}{T}\right \rfloor, u\in \mathbb{Z}_m^k.$$ 
The following theorem gives a value of $\lambda$ for which the Lee weight distribution function $W_L^T$ is a locally $(\lambda,\rho)_L$-function for an arbitrary positive integer $\rho$.
\begin{theorem}\label{thmlwd}
  Let $\rho$  be a positive integer. Then the Lee weight distribution function $W_L^T$ over $\mathbb{Z}_m^k$ is a locally $\left (\left \lfloor\frac{2\rho}{T}\right \rfloor + 2,\rho\right )_L$-function.
\end{theorem}
\begin{proof}
Let $f(u)=W_L^T(u)$. Let $u\in \mathbb{Z}_m^k$, by Equation \ref{leew}, for any $v\in B_L(u,\rho)$, we have 
$$ \left \lfloor \frac{w_L(u)-\rho}{T}\right \rfloor \leq f(v)= \left \lfloor \frac{w_L(v)}{T}\right \rfloor\leq \left \lfloor \frac{w_L(u)+\rho}{T}\right \rfloor.$$
Therefore, $$\mid B_L^f(u,\rho)\mid\leq \left \lfloor \frac{w_L(u)+\rho}{T}\right \rfloor-\left \lfloor \frac{w_L(u)-\rho}{T}\right \rfloor+1\leq \left \lfloor \frac{(w_L(u)+\rho)- (w_L(u)-\rho)}{T}\right \rfloor+2,$$
that is, $ |B_L^f(u,\rho)|\leq \left \lfloor \frac{2\rho}{T}\right \rfloor+2$. Since $u$ is arbitrary, therefore, $W_L^T$ is a locally $\left (\left \lfloor\frac{2\rho}{T}\right \rfloor + 2,\rho\right )_L$-function.
\end{proof}

\begin{remark}
In view of  Theorem \ref{thmlwd}, let $T>0$  and $\rho=2t$ for some positive integer $t$. We observe that:
\begin{enumerate}
    \item If $T> 4t$, then the Lee weight distribution function $W_L^T$ is a locally $(2,2t)_L$-function.
    \item If $\frac{4t}{2}< T\leq 4t$, then the Lee weight distribution function $W_L^T$ is a locally $(3,2t)_L$-function. 
    \item In general, let $\lambda >2$ be an integer. If $\frac{4t}{\lambda-1}< T\leq \frac{4t}{\lambda-2}$, then the Lee weight distribution function $W_L^T$ is a locally $(\lambda,2t)_L$-function.
\end{enumerate} 
Thus, the Lee weight distribution function $W_L^T$ is a locally $(\lambda,2t)_L$-function satisfying the contiguous block condition in Lemma \ref{tau} for all possible choices of $\lambda$ depending on threshold $T$ and $t$.
\end{remark}

\begin{corollary}
    The optimal redundancy of an $(f,t)$-FCLMC for the Lee weight distribution function $f(u)=W_L^T(u)$ satisfies
    \begin{align*}
   r_f^L(k,t)=\begin{cases}
       0 & \text{ if } T> k\left \lfloor\frac{m}{2}\right \rfloor\\
       \left \lceil\frac{2t}{\left \lfloor \frac{m}{2}\right \rfloor} \right \rceil & \text{ if } T= k\left \lfloor\frac{m}{2}\right \rfloor,
   \end{cases}\\ 
\text{and for }  T<k\left \lfloor\frac{m}{2}\right \rfloor, \text{ we have }\hfill\\        r_f^L(k,t)\leq \left(\left \lfloor\frac{4t}{T}\right \rfloor + 2\right)\left \lceil \frac{t}{\left \lfloor \frac{m}{2}\right \rfloor}  \right\rceil.
         \end{align*}
Moreover, Proposition \ref{prop-l2t} provides a construction of FCLMC achieving the above bound.         
\end{corollary}
\begin{proof}
  If $T> k\left \lfloor\frac{m}{2}\right \rfloor$, then $W_L^T(u)=0$ for all $u\in \mathbb{Z}_m^k$, that is, $W_L^T$ is a constant function. Hence, no redundancy is required.  If $T= k\left \lfloor\frac{m}{2}\right \rfloor$, then $W_L^T(u)\in \{0,1\}$, that is, $W_L^T$ is $(2,2t)_L$-function. By Corollary \ref{cor-redundancy 2-locally}, we get the desired result.
 For $T< k\left \lfloor\frac{m}{2}\right \rfloor$, by Theorem  \ref{thmlwd}, the Lee weight distribution function $W_L^T(u)$ is a locally $\left (\left \lfloor \frac{4t}{T}\right \rfloor+2,2t\right)
    _L$-function. The rest follows from Theorem \ref{thmN2t} and Proposition \ref{prop-l2t}.
\end{proof}


\section{Plotkin-like bound for linear function-correcting codes in the Lee metric}\label{linearfcc}
In this section, we derive a Plotkin-like bound for a specific class of functions, namely linear functions, in the context of channels measured using the Lee metric. Linear function correcting codes over the Hamming metric are discussed in \cite{Premlal2024}, and later extended to $b$-symbol read channels in \cite{Sampath2025}.

\begin{definition}
A function $f:\mathbb{Z}_m^k\to \mathbb{Z}_m^\ell$ is said to be a linear function if $f(\alpha x+y)=\alpha f(x)+f(y)$ for all $\alpha\in \mathbb{Z}_m$ and $x,y\in \mathbb{Z}_m^k$. The kernel of a liner function $f$ is defined as 
$\ker(f)=\{x\in \mathbb{Z}_m^k:\ f(x)=\bm{0}\}$.
\end{definition}

\begin{theorem}\label{plotlinear}
    Let $f:\mathbb{Z}_m^k\to \mathbb{Z}_m^\ell$ be an onto linear function. Then the optimal redundancy of an $(f,t)$-FCLMC is bounded by 
    \begin{align*}
        r_f^L(k,t)\geq \frac{m}{S}(2t+1)(1-m^{-\ell})-k+\frac{s}{Sm^{k-1}},
    \end{align*}
    where $S$ is as in Lemma \ref{sumS} and $s=\sum_{u\in \ker(f)}w_L(u)$.
\end{theorem}
\begin{proof}
  Let $M=m^k$, $\mathbb{Z}_m^k=\{u_1,u_2,\dots,u_M\}$, and $D=D_f^L(t,u_1,u_2,\dots,u_M)$. Let $r=N_L(D)=r_f^L(k,t)$ (by Theorem \ref{rfnf}), $C=\{C_1,C_2,\dots,C_M\}\subseteq \mathbb{Z}_m^r$ be a $D$-code in the Lee metric of length $r$, that is, $d_L(C_i,C_j)\geq [D]_{ij}$. Then, by Lemma \ref{Csum}, we have 
\begin{equation}\label{eq7.1}
      \frac{SM^2r}{m}=\sum_{i,j}d_L(C_i,C_j)\geq \sum_{i,j}[D]_{ij}.
  \end{equation}
Since $f$ is an onto linear function, the cardinality of $\ker(f)$ is $m^{k-\ell}$. Thus, each coset of $\ker(f)$ in $\mathbb{Z}_m^k$ contains $m^{k-\ell}$ elements. Observe that if $u_i,u_j$ belong to the same coset, then $f(u_i)=f(u_j)$. Consequently, $[D]_{ij}=0$, that is, each column of $D$ contains at least $m^{k-\ell}$ zeros. Let $Col_a$ and $Col_b$ be two distinct columns in $D$ corresponding to vectors $u_a$ and $u_b$, respectively. Let $[D]_{ia}$ be a nonzero entry in  $Col_a$ corresponding to vector $u_i$, that is, $[D]_{ia}=2t+1-d_L(u_i,u_a)$. Let $u_j=u_i-u_a+u_b\in \mathbb{Z}_m^k$. Then the entry $[D]_{jb}$ in $Col_b$ corresponding to $u_j$ is $[D]_{jb}=2t+1-d_L(u_j,u_b)=2t+1-d_L(u_i,u_a)=[D]_{ia}$, that is, the $[D]_{ia}$ entry in $Col_a$ is same as $[D]_{jb}$ entry in $Col_b$. Therefore, each column in $D$ contains the same set of entries, implying that the columns are permutations of one another. Thus, 
\begin{align*}
  \sum_{i,j} [D]_{ij}=(\text{no. of columns)}\times (\text{sum of one column)}.  
\end{align*}
Let $u_1=\bm{0}$ be the zero vector in $\mathbb{Z}_m^k$. We calculate the sum of the entries of the column corresponding to $u_1$ which is the first column of $D$. If $u_i\in \ker(f)$, then $[D]_{i1}=0$. Let $I$ be the row indices of the non-zero entries in the first column of $D$. For $i\in I$, $[D]_{i1}\geq 2t+1-d_L(u_i,\bm{0})=2t+1-w_L(u_i)$. Thus,
\begin{align*}
   \sum_{i}[D]_{i1}\geq (2t+1)(m^k-m^{k-\ell})-\sum_{i\in I} w_L(u_i). 
\end{align*}
The sum of  Lee weights of all the vectors in $\mathbb{Z}_m^k$ is $kSm^{k-1}$. Thus, \begin{align*}
    \sum_{i\in I}w_L(u_i)=kSm^{k-1}-s, \text{ where } s=\sum_{u\in \ker(f)}w_L(u). 
\end{align*} 
Therefore,
\begin{align}\label{eq7.2}
\begin{split}
    \sum_{i,j} [D]_{ij}&\geq M\cdot  \sum_{i}[D]_{i1}\\
      &\geq m^k((2t+1)(m^k-m^{k-\ell})-kSm^{k-1}+s).
    \end{split}
\end{align}
The desired result follows from Eqs. \ref{eq7.1} and \ref{eq7.2}. 
\end{proof}

As stated in Lemma \ref{plotkin}, the Plotkin bound depends on the sum of all $[D]_{ij}$ terms. While in the case of linear functions, this reduces to the sum of the Lee weights of the vectors in $\ker(f)$. In \cite{Premlal2024}, the authors derived a Plotkin-like bound for linear function-correcting codes over the finite field $\mathbb{F}_q$ under the Hamming metric.  As the Lee metric coincides with the Hamming metric over $\mathbb{Z}_2$, our corresponding result in the Hamming metric over $\mathbb{Z}_2$ is consistent with the bound in \cite{Premlal2024}.
\begin{corollary}
 Let $f:\mathbb{Z}_2^k\to \mathbb{Z}_2^\ell$ be an onto linear function. Then the optimal redundancy of an $(f,t)$-FCC is bounded by 
    \begin{align*}
        r_f^H(k,t)\geq 2(2t+1)(1-2^{-\ell})-k+\frac{s}{2^{k-1}},
    \end{align*}
    where $s=\sum_{u\in \ker(f)}w_H(u)$.    
\end{corollary}
\begin{proof}
By Lemma \ref{sumS}, we have $S=1$ for $m=2$. The rest follows from Theorem \ref{plotlinear} by putting $m=2$.    
\end{proof}

 When the function $f$ is bijective, a function-correcting code reduces to a systematic error-correcting code. Particularly, setting $f$ to be a bijective function, that is, $s=0$ and $\ell=k$, the bound presented in Theorem~\ref{plotlinear} simplifies to the Poltkin bound for error-correcting codes in the Lee metric. This aligns precisely with the well-known  Plotkin low-rate average distance bound (see \cite{Wyner1968}). 
 \begin{corollary}
     Let $f:\mathbb{Z}_m^k\to \mathbb{Z}_m^k$ be a bijective linear function. Then the optimal redundancy of an $(f,t)$-FCLMC is bounded by 
  \begin{align*}
     n=r_f^L(k,t)+k\geq \frac{m}{S}(2t+1)\left(\frac{m^k-1}{m^k}\right),
 \end{align*}
    where $S$ is as in Lemma \ref{sumS}.
 \end{corollary}
The following corollary derives the classical Plotkin bound for error-correcting codes over $\mathbb{Z}_2$ in the Hamming metric.
\begin{corollary}
     Let $f:\mathbb{Z}_2^k\to \mathbb{Z}_2^k$ be a bijective linear function. Then the optimal redundancy of an $(f,t)$-FCC is bounded by 
  \begin{align*}
     n=r_f^H(k,t)+k\geq (2t+1)\left(\frac{2^k-1}{2^{k-1}}\right).
 \end{align*}
 \end{corollary}




\section{Conclusion}
 In this paper, we extended the study of function-correction codes for the Lee metric by introducing function-correcting Lee metric codes (FCLMCs) over the residue ring $\mathbb{Z}_m$. FCLMCs require less redundancy compared to classical Lee metric codes. In this encoding, we only need to distinguish those codewords whose associated function values differ, rather than correcting all erroneous codewords. We derived upper and lower bounds on the optimal redundancy of FCLMCs for general functions utilizing the smallest length irregular distance codes in the Lee metric. Moreover, we defined locally bounded functions in the Lee metric, and using general results, we obtained bounds on the optimal redundancy for locally bounded functions. We also showed that the lower bound on optimal redundancy can be attained in the case of locally binary functions.  We further explored the locality of specific functions, namely the Lee weight function and the Lee weight distribution function, and established bounds on the optimal redundancy. Moreover, we provided explicit constructions for these functions. Our results generalize the corresponding bounds over $\mathbb{Z}_2$ given in \cite{Lenz2023} and over $\mathbb{Z}_4$ given in \cite{liu2025}. Additionally, we explored linear FCLMCs and established a Plotkin-like bound for linear function-correcting codes in the Lee metric.

 A natural and nontrivial continuation of this work is the investigation and construction of coding schemes that provide simultaneous protection against errors in both the function outputs and the underlying message data. Developing such codes, capable of jointly safeguarding data and function values in the Lee metric, constitutes a challenging and promising direction for future research, with potential impact on error-resilient computation, storage, and communication systems. Recently, in \cite{Rajput2025fcc-data}, the authors studied such code constructions in the Hamming metric.

\section*{Acknowledgment}
The authors are thankful to the anonymous reviewers and the
editor for their valuable suggestions that improve the overall quality of the paper.

\bibliographystyle{abbrv}
	\bibliography{ref}

\end{document}